\documentclass{article}
\usepackage{amsfonts}
\usepackage{bbm}
\usepackage{amsmath, amsthm, amssymb}
\usepackage{color,calc}
\usepackage{hyperref}

\newcommand{\sech}{\operatorname{sech}}

\newcommand{\D}{\Delta}
\newcommand{\G}{\Gamma}

\newcommand{\Wr}{\mathrm{Wrd}}

\newcommand{\af}{\alpha}
\newcommand{\dt}{\delta}
\newcommand{\ta}{\tau}
\newcommand{\kp}{\kappa}
\newcommand{\et}{\eta}

\newcommand{\tht}{\theta}

\newcommand{\Om}{\Omega}
\newcommand{\ld}{\lambda}

\newtheorem{prop}{Proposition}
\newtheorem{lemm}{Lemma}

\newtheorem{exam}{Example}

\begin{document}

\begin{center}{\Large \bf A new extended discrete KP hierarchy and generalized dressing method}
\end{center}
\begin{center}
{\it Yuqin Yao$^{1)}$\footnote{yqyao@math.tsinghua.edu.cn }, Xiaojun
Liu$^{2)}$\footnote{tigertooth4@gmail.com} and Yunbo
Zeng$^{1)}$\footnote{Corresponding author:

~~~~yzeng@math.tsinghua.edu.cn} }
\end{center}
\begin{center}{\small \it $^{1)}$Department of Mathematical Science,
Tsinghua University, Beijing, 100084 , PR China\\
$^{2)}$Department of
  Applied Mathematics, China Agricultural University, Beijing, 100083, PR China}
\end{center}

\vskip 12pt { \small\noindent\bf Abstract}
 {Inspired by
the squared eigenfunction symmetry constraint, we introduce a new
$\ta_k$-flow by ``extending'' a specific $t_n$-flow of discrete KP
hierarchy (DKPH). We construct extended discrete KPH (exDKPH), which
consists of
 $t_n$-flow, $\ta_k$-flow  and $t_n$ evolution of eigenfunction and
 adjoint eigenfunctions, and its Lax representation.
The exDKPH contains
 two types of discrete KP equation with self-consistent sources (DKPESCS).
  Two reductions of exDKPH are obtained. The generalized dressing approach
  for solving the exDKPH is proposed and the N-soliton solutions of two types of the DKPESCS are
  presented.}

\section{Introduction}
Generalizations of soliton hierarchy attract a lot of interests from
both physical and mathematical points and there were some methods to
generalize the soliton hierarchy\cite{g1}-\cite{g7}.  Recently, a
systematic approach inspired by squared eigenfunction symmetry
constraint was proposed to construct the extended KP
hierarchy\cite{g10}. By this method, the extended two-dimensional
Toda lattice hierarchy, the extended CKP hierarchy and the extended
q-deformed KP hierarchy have been obtained\cite{g11}-\cite{g13}.

 The discrete KP hierarchy(DKPH) \cite{DE}-\cite{tau} is an
interesting object in the research of the discrete integrable
systems and the discretization of the integrable
systems\cite{discre1}. The Sato's approach for the discrete KPH was
presented in \cite{VT}. Naturally, there are some similar properties
between discrete KPH and KPH\cite{kp}, such as tau
function\cite{tau,kp}, Hamiltonian structure\cite{tau} and gauge
transformation\cite{dkp1,gt1,dkp5}, etc. In \cite{dkp1}, Oevel has
given explicitly two types of gauge transformation operators of the
discrete KPH. In \cite{dkp5}, the combined gauge operator and the
determinant representation of the operator have been obtained.

In this paper, we will construct the extension of the discrete
KPH(exDKPH). Inspired by the squared eigenfunction symmetry
constraint of discrete KP hierarchy \cite{dkp1}, we introduced the
new $\ta_k$-flow by ``extending'' a specific $t_n$-flow of discrete
KP hierarchy. Then we find the exDKPH consisting of $t_n$-flow of
discrete KP hierarchy, $\ta_k$-flow and the $t_n$-evolutions of
eigenfunctions and adjoint eigenfunctions. The commutativity of
$t_n$-flow and $\ta_k$-flow gives rise to zero curvature
representation for exDKPH. Also the Lax representation of exDKPH is
derived. Due to the introduction of $\ta_k$-flow the exDKPH contains
two time series $\{t_n\}$ and $\{\ta_k\}$ and more components by
adding eigenfunctions and adjoint eigenfunctions. The exDKPH
contains the first type and second type of discrete KP equation with
self-consistent sources(DKPESCS). The KP equation with
self-consistent sources arose in some physical models describing the
interaction of long and short waves\cite{g7}. The similarity of KP
equation and discrete KP equation enables us to speculate on the
potential application of discrete KP equation with self-consistent
sources. By $t_n$-reduction and $\ta_k$-reduction, the exDKPH
reduces to a discrete 1+1-dimensional integrable hierarchy with
self-consistent sources and constrained discrete KP hierarchy,
respectively.

The dressing method is an important tool for solving soliton
hierarchy \cite{tau}. However this method can not be applied
directly for solving the ``extended'' hierarchy. A generalized
dressing approach for exKPH is proposed in \cite{dress}. In this
paper, with the combination of dressing method and variation of
constants method, a generalization to the dressing method for exDKPH
is presented, which is based on the dressing method for discrete KPH
\cite{VT} and the similar approach for finding Wronskian solutions
to constrained KP hierarchy \cite{OS96}. In this way, we can solve
the entire hierarchy of exDKPH in an unified and simple manner. As
the special cases, the N-soliton solutions of the \emph{both} types
of DKPESCS
 are obtained
simultaneously.

 This paper will be organized as follows. In Sec.2, we present the
exDKPH and its Lax pair, which includes two types of  DKPESCS. In
Sec.3, $t_n$-reduction and $\tau_k$- reduction for the exDKPH are
given. In Sec.4, we discuss the generalized dressing method for the
exDKPH. In Sec.5, we present the N-soliton solutions of the DKPESCS.

\section{New extended discrete KP hierarchy}
We denote  the shift and the difference operators acting on the
associative ring $F$ of functions by $\Gamma$ and $\Delta$,
respectively, as follows
$$F=\{f(l)=f(l,t_1,t_2,\cdots,t_i,\cdots); l\in \mathbb{Z},~t_i\in\mathbb{R} \}$$
$$\Gamma (f(l))=f(l+1)= f^{(1)}(l),~\Delta (f(l))=f(l+1)-f(l).$$
In this paper, we use $P(f)$ to denote an action of difference
operator $P$ on the function $f$, while $Pf$ means the
multiplication of difference operator $P$ and zero order difference
operator $f$. Define the following operation
\begin{equation}
 \label{eqn:def1}
\small{\Delta^{j} f=\sum\limits_{i=0}^{\infty}\left(\begin{array}{c}
    j\\
    i\\
  \end{array}\right)(\Delta^{i}(f(l+j-i)))\Delta^{j-i},~\left(\begin{array}{c}
    j\\
    i\\
  \end{array}\right)=\frac{j(j-1)\cdots(j-i+1)}{i!}.}\end{equation}
Also, we define the adjoint operator to the $\Delta$ operator by
$\Delta^{*}$
\begin{equation}
 \label{eqn:2}
 \Delta^{*}(f(l))=(\Gamma^{-1}-I)(f(l))=f(l-1)-f(l),\end{equation}
\begin{equation}
 \label{eqn:def2}
\Delta^{*j}f=\sum\limits_{i=0}^{\infty}\left(\begin{array}{c}
    j\\
    i\\
  \end{array}\right)(\Delta^{*i}(f(l+i-j)))\Delta^{*j-i}.\end{equation}
  Let $P=\sum_{j=-\infty}^{k}f_{j}(l)\Delta^{j}$, the adjoint
  operator $P^{*}$ is defined by $P^{*}=\sum_{j=-\infty}^{k}\Delta^{*j}f_{j}(l).$

The Lax equation of the DKP hierarchy is given by\cite{DE,VT}
\begin{equation}
  \label{eqn:dKP-LaxEqn}
  L_{t_n}=[B_n,L],
\end{equation}
where  $L = \D + f_0 + f_1\D^{-1}+f_2\D^{-2}+\cdots $ is a
pseudo-difference operator with potential functions $f_i\in F$,
$B_n=L^n_{+}$ stands for the difference part of $L^{n}$. The
commutativity of $t_n$- and $t_m$-flow gives rise to the
zero-curvature equations for DKP hierarchy:
\begin{equation}
 \label{eqn:zeroc}
  B_{n,t_m}-B_{m,t_n}+[B_n,B_m]=0.
\end{equation}
with the  Lax pair given by
\begin{equation}
  \label{eqn:dKP-LaxPair}
  \psi_{t_n}=B_n(\psi),\quad\psi_{t_m}=B_m(\psi).
\end{equation}
The $t_n$ evolutions of eigenfunction $\psi$ and adjoint
eigenfunction $\phi$ read
\begin{equation}
\label{eqn:dKP-aLaxPair}
  \psi_{t_n}=B_n(\psi),\quad \phi_{t_n}=-B_n^*(\phi).
\end{equation}
For $n=2,~m=1,$ (\ref{eqn:zeroc}) gives rise to the DKP
equation\cite{DE}
\begin{equation}
\label{eqn:dKPe}
 \Delta(f_{0t_{2}}+2f_{0t_{1}}-2f_{0}f_{0t_{1}})=(\Delta+2)f_{0t_{1}t_{1}}.
\end{equation}

 It is known that the
squared eigenfunction symmetry constraint given by\cite{dkp1}
$$
   \tilde{B}_{k}=B_k+\sum_{i=1}^N\psi_i\D^{-1}\phi_i~~~~$$$$
    \psi_{i,t_n}=B_n(\psi_i),~
    \phi_{i,t_n}=-B_n^*(\phi_i),~i=1,\cdots,N,
 $$
is compatible with  DKP hierarchy. Here $N$ is an arbitrary natural
number, $\psi_i$ and $\phi_i$ are $N$ different eigenfunctions and
adjoint eigenfunctions of the equations (9c). This compatibility
enables us to construct a new extended discrete KP hierarchy
(exDKPH) as
\begin{subequations}
  \label{eqns:exdKP-LaxEqn}
  \begin{align}
    L_{t_n}&=[B_n,L],  \label{eqn:exdKP-LaxEqna}\\
    L_{\ta_k}&=[B_k+\sum_{i=1}^N\psi_i\D^{-1}\phi_i,L], \label{eqn:exdKP-LaxEqnb}\\
    \psi_{i,t_n}&=B_n(\psi_i),~
    \phi_{i,t_n}=-B_n^*(\phi_i),~i=1,\cdots,N. \label{eqn:exdKP-LaxEqnc}
  \end{align}
\end{subequations}
We have the following lemma.
\begin{lemm}
  \label{lm:1}
  Let $Q=a\Delta^{k},~k\geq1,$ then
\begin{subequations}
  \label{eqns:lemma1}
  \begin{align}
    (\Delta^{-1}\phi Q)_{-}& = \Delta^{-1}Q^{*}(\phi)\label{eqn:lemma1a}\\
    [B_n,\psi\Delta^{-1}\phi]_{-}&=B_n(\psi)\Delta^{-1}\phi-\psi\Delta^{-1}B^{*}_n(\phi). \label{eqn:lemma1b}
      \end{align}
\end{subequations}
\end{lemm}
\begin{proof}
Using
$f\Delta=\Delta\Gamma^{-1}(f)-\Delta(\Gamma^{-1}(f)),~\Delta^{*}=-\Delta\Gamma^{-1}$
,  we have
$$(\Delta^{-1}\phi
a\Delta^{k})_{-}=(\Delta^{-1}\Delta\Gamma^{-1}(\phi a)\Delta^{k-1}-
\Delta^{-1}\Delta(\Gamma^{-1}(\phi
a))\Delta^{k-1})_{-}$$$$=-(\Delta^{-1}\Delta(\Gamma^{-1}(\phi
a))\Delta^{k-1})_{-}=\cdots$$
$$=(-1)^{k}
\Delta^{-1}\Delta^{k}(\Gamma^{-k}(\phi
a))=\Delta^{-1}\Delta^{*k}(\phi a) = \Delta^{-1}Q^{*}(\phi)$$ which
yields to (\ref{eqn:lemma1a}) and (\ref{eqn:lemma1b}).
\end{proof}
\begin{prop}
  \label{prop:1}
  The commutativity of (\ref{eqn:exdKP-LaxEqna}) and (\ref{eqn:exdKP-LaxEqnb}) under (\ref{eqn:exdKP-LaxEqnc})
   gives rise to the following zero-curvature representation for
   exDKPH (\ref{eqns:exdKP-LaxEqn})
  \begin{subequations}
    \label{eqns:extdkp}
    \begin{align}
      B_{n,\tau_{k}}&-(B_{k}+\sum\limits_{i=1}^{N}\psi_{i}\Delta^{-1}\phi_{i})_{t_{n}}+[B_{n},B_{k}
      +\sum\limits_{i=1}^{N}\psi_{i}\Delta^{-1}\phi_{i}]=0,\label{eqn:extdkpa}\\
      \psi_{i,t_n}&=B_n(\psi_i),~
      \phi_{i,t_n}=-B_n^*(\phi_i),~i=1,2,\cdots,N,\label{eqn:extdkpb}
    \end{align}
  \end{subequations}
  with the  Lax representation given by
\begin{equation}
  \label{eqn:exdKP-LaxPair}
  \Psi_{t_n}=B_n(\Psi),\quad \Psi_{\ta_k}=(B_k+\sum_{i=1}^N\psi_i\D^{-1}\phi_i)(\Psi).
\end{equation}
\end{prop}
\begin{proof}
  For convenience, we omit $\sum$.
  By (\ref{eqns:exdKP-LaxEqn}) and Lemma 1, we have
$$ B_{n,\tau_{k}}=(L^{n}_{\tau_{k}})_{+}=[B_{k}+\psi\Delta^{-1}\phi,L^{n}]_{+}
=[B_{k}+\psi\Delta^{-1}\phi,L^{n}_{+}]_{+}+[B_{k}+\psi\Delta^{-1}\phi,L^{n}_{-}]_{+}
$$
$$=[B_{k}+\psi\Delta^{-1}\phi,L^{n}_{+}]
-[B_{k}+\psi\Delta^{-1}\phi,L^{n}_{+}]_{-}+[B_{k},L^{n}_{-}]_{+}=[B_{k}+\psi\Delta^{-1}\phi,L^{n}_{+}]~~
$$
$$-[\psi\Delta^{-1}\phi,B_{n}]_{-}
+[B_{n},L^{k}]_{+}=[B_{k}+\psi\Delta^{-1}\phi,B_{n}]+(B_{k}+\psi\Delta^{-1}\phi)_{t_{n}}.~~~~~~~~~~~$$
\end{proof}
{\bf Remark.} The exDKPH (11) extends the DKPH (5) by containing two
time series $\{t_{n}\}$ and $\{\tau_{k}\}$ and more components
$\psi_i$ and $\phi_i$, $i=1,\cdots,N.$
\begin{exam}
  The first type of DKPSCS is given by  (\ref{eqns:extdkp}) with  $n=1,~k=2$
  \begin{subequations}
    \label{eqns:exam1}
    \begin{align}
      \Delta
      (f_{0\tau_{2}}&+2f_{0t_{1}}-2f_{0}f_{0t_{1}})=(\Delta+2)f_{0t_{1}t_{1}}-
      \Delta^{2}\sum\limits_{i=1}^{N}(\psi_{i}\phi_{i}^{(-1)}),\label{eqns:exam1a}\\
      \psi_{i,t_1}&=\Delta(\psi_{i})+f_{0}\psi_{i},~\phi_{i,t_1}=-\Delta^{*}(\phi_{i})-f_{0}\phi_{i},~i=1,2,\cdots,N. \label{eqns:exam1b}
         \end{align}
  \end{subequations}
   Its Lax
  representation is
  \begin{subequations}
   \label{eqns:dkp1-lax}
    \begin{align}
      \Psi_{t_1}&=(\Delta+f_{0})(\Psi)\\
      \Psi_{\tau_{2}}&=(\Delta^{2}+(f_{0}+f^{(1)}_{0})\Delta+\Delta(f_{0})+f^{(1)}_{1}
      +f_{1}+f^{2}_{0}+\sum\limits_{i=1}^{N}\psi_{i}\Delta^{-1}\phi_{i})(\Psi).
    \end{align}
 \end{subequations}
\end{exam}
\begin{exam}
  The second type of DKPSCS is given by (\ref{eqns:extdkp}) with  $n=2,~k=1$
  \begin{subequations}
    \label{eqns:exam2}
    \begin{align}
      \Delta
      (f_{0t_{2}}&+2f_{0\tau_{1}}-2f_{0}f_{0\tau_{1}})=(\Delta+2)f_{0\tau_{1}\tau_{1}}+
   \sum\limits_{i=1}^{N}[\Delta^{2}((f_{0}+f^{-1}_{0}-2)\psi_{i}\phi_{i}^{-1})\nonumber \\
      &~~~~~~~~~~~~~~~~~~~+\Delta(\psi_{i}^{(2)}\phi_{i}-\psi_{i}\phi_{i}^{(-2)})+
      \Delta((\Gamma+1)(\psi_{i}\phi_{i}^{(-1)})_{\tau_{1}})],\\
      \psi_{i,t_2}&=\Delta^{2}(\psi_{i})+(f_{0}+f^{(1)}_{0})\Delta(\psi_{i})+(\Delta(f_{0})+f^{(1)}_{1}
      +f_{1}+f^{2}_{0})\psi_{i},\\
      \phi_{i,t_2}&=-\Delta^{*2}(\psi_{i})-\Delta^{*}((f_{0}+f^{(1)}_{0})\psi_{i})-(\Delta(f_{0})+f^{(1)}_{1}
      +f_{1}+f^{2}_{0})\psi_{i}.
    \end{align}
  \end{subequations}
  Its Lax
 representation is
  \begin{subequations}
    \label{eqns:dkp1-lax}
    \begin{align}
     \Psi_{t_2}&=(\Delta^{2}+(f_{0}+f^{(1)}_{0})\Delta+\Delta(f_{0})+f^{(1)}_{1}
      +f_{1}+f^{2}_{0})(\Psi)\\
      \Psi_{\tau_{1}}&=(\Delta+f_{0}+\sum\limits_{i=1}^{N}\psi_{i}\Delta^{-1}\phi_{i})(\Psi).
    \end{align}
  \end{subequations}
\end{exam}

\section{Reductions of the exDKPH}
\subsection{The $t_n$- reduction}
The $t_n$-reduction is given by
\begin{equation}
  \label{eqn:tnr}
  L^{n}=B_{n} ~~~or ~~~L_{-}^{n}=0.
\end{equation}
Then we have
\begin{displaymath}
  (L^n)_{t_n}=[B_n, L^n]=0,\quad B_{n,t_n}=0.
\end{displaymath}
So $L$ is independent of $t_n$ and we have
\begin{equation}
    \label{eqn:eigenfuction}
       B_n(\psi_i)=L^n(\psi_i)=\lambda_i^n\psi_i,~
      B_n^*(\phi_i)=\lambda_i^n\phi_i.
      \end{equation}
Then we can drop $t_n$ dependency from (\ref{eqns:extdkp}) and
obtain
 \begin{subequations}
    \label{eqns:extdkptnr}
    \begin{align}
      B_{n,\tau_{k}}&=[(B_{n})^{\frac{k}{n}}_{+}
      +\sum\limits_{i=1}^{N}\psi_{i}\Delta^{-1}\phi_{i},B_{n}], \label{eqns:extdkptnra}\\
      B_n(\psi_i)&=\lambda_{i}^{n}\psi_i,~
   B_n^*(\phi_i)=\lambda_{i}^{n}\phi_i,~i=1,2,\cdots,N, \label{eqns:extdkptnrb}
    \end{align}
  \end{subequations}
  with the Lax pair given by
  $$\Psi_{\tau_{k}}=((B_{n})^{\frac{k}{n}}_{+}+\sum_{i=1}^{N}\psi_{i}\Delta^{-1}\phi_{i})(\Psi),~B_{n}(\Psi)=\lambda^{n}\Psi.$$
(\ref{eqns:extdkptnr}) can be regarded as  discrete
(1+1)-dimensional integrable hierarchy with self-consistent sources.
When $n=2,~k=1$, (\ref{eqns:extdkptnr}) gives rise to
\begin{subequations}
    \label{eqns:exam2tnr}
    \begin{align}
      2\Delta
      (&f_{0\tau_{1}}-f_{0}f_{0\tau_{1}})=(\Delta+2)f_{0\tau_{1}\tau_{1}}+
   \sum\limits_{i=1}^{N}[\Delta^{2}(f_{0}+f^{(-1)}_{0}-2)\psi_{i}\phi_{i}^{-1}\nonumber \\
      &~~~~~~~~~+\Delta(\psi_{i}^{(2)}\phi_{i}-\psi_{i}\phi_{i}^{(-2)})+
      \Delta(\Gamma+1)(\psi_{i}\phi_{i}^{(-1)})_{\tau_{1}}]\\
      \Delta^{2}&(\psi_{i})+(f_{0}+f^{(1)}_{0})\Delta(\psi_{i})+(\Delta(f_{0})+f^{(1)}_{1}
      +f_{1}+f^{2}_{0})\psi_{i}=\lambda_{i}^{2}\psi_i,\\
      \Delta^{*2}&(\psi_{i})+\Delta^{*}((f_{0}+f^{(1)}_{0})\psi_{i})+(\Delta(f_{0})+f^{(1)}_{1}
      +f_{1}+f^{2}_{0})\psi_{i}=\lambda_{i}^{2}\phi_i,
    \end{align}
  \end{subequations}
which can be transformed  to the first type of Veselov-Shabat
equation\cite{vs} with self-consistent sources (VSESCS).

\subsection{The $\tau_k$- reduction}
The $\tau_{k}$-reduction is given by\cite{dkp1}
$$L^{k}=B_{k}+\sum\limits_{i=1}^{N}\psi_{i}\Delta^{-1}\phi_{i}.$$
By dropping $\tau_k$ dependency from (\ref{eqns:extdkp}), we obtain
 \begin{subequations}
    \label{eqns:extdkptkr}
    \begin{align}
    (B_{k}&+\sum\limits_{i=1}^{N}\psi_{i}\Delta^{-1}\phi_{i})_{t_{n}}=[
    (B_{k}+\sum\limits_{i=1}^{N}\psi_{i}\Delta^{-1}\phi_{i})^{\frac{n}{k}}_{+},
    B_{k}+\sum\limits_{i=1}^{N}\psi_{i}\Delta^{-1}\phi_{i}],\\
      \psi_{i,t_{n}}&=(B_{k}+\sum\limits_{i=1}^{N}\psi_{i}\Delta^{-1}\phi_{i})_{+}^{\frac{n}{k}}(\psi_{i}),\\
  \phi_{i,t_{n}}&=-(B_{k}+\sum\limits_{i=1}^{N}\psi_{i}\Delta^{-1}\phi_{i})_{+}^{\frac{n}{k}*}(\phi_{i}),~i=1,2,\cdots,N,
    \end{align}
  \end{subequations}
which is the k-constrained DKP hierarchy. When $n=1,~k=2$,
(\ref{eqns:extdkptkr}) leads to
\begin{subequations}
    \label{eqns:exam1tkr}
    \begin{align}
      2\Delta
      (&f_{0t_{1}}-f_{0}f_{0t_{1}})=(\Delta+2)f_{0t_{1}t_{1}}+
      \Delta^{2}\sum\limits_{i=1}^{N}(\psi_{i}\phi_{i}^{(-1)}),\\
      \psi_{i,t_1}&=\Delta(\psi_{i})+f_{0}\psi_{i},~
      \phi_{i,t_1}=-\Delta^{*}(\phi_{i})-f_{0}\phi_{i},~i=1,2,\cdots,N,
    \end{align}
  \end{subequations}
which
can be transformed  to the second type of VSESCS.

\section{Dressing approach for exDKPH}

\subsection{Dressing approach  for discrete KP hierarchy}
We first briefly recall the dressing approach for DKPH \cite{VT}.
Assume that operator $L$ of DKPH (\ref{eqn:dKP-LaxEqn}) can be
written as a dressing form
\begin{equation}
  \label{eqn:dress}
  L=W\Delta W^{-1},
\end{equation}
\begin{displaymath}
  W = \D^{N} + w_1\D^{N-1} + w_2\D^{N-2}+\cdots + w_N.
\end{displaymath}
It is known
  \cite{tau} that if $W$ satisfies
\begin{equation}
  \label{eqn:W-evolution-dKP}
  W_{t_n}=-L^{n}_-W,
\end{equation}
then  $L$ satisfies (\ref{eqn:dKP-LaxEqn}). It is easy to check the
following Lemma.
\begin{lemm}
\label{eqn:add} If $h_{t_n}=\D^n(h)$, $W$ satisfies
(\ref{eqn:W-evolution-dKP}), then $\psi = W(h)$ satisfies
(\ref{eqn:dKP-aLaxPair}), i.e.
\begin{equation}
  \label{eqn:adlp}
  \psi_{t_n}=B_{n}(\psi).
\end{equation}
\end{lemm}
If there are $N$ independent functions $h_1,\ldots, h_N$ solving
$W(h)=0$ i.e. $W(h_{i})=0$, then $w_1,\ldots, w_N$ are completely
determined from these $h_i$, by solving the linear equation:
\begin{displaymath}
  \begin{pmatrix}
    h_1 & \D(h_1) & \cdots & \D^{N-1}(h_1)\\
    h_2 & \D(h_2) & \cdots & \D^{N-1}(h_2)\\
    \vdots & \vdots & \vdots & \vdots\\
    h_N & \D(h_N) & \cdots & \D^{N-1}(h_N)
  \end{pmatrix}
  \begin{pmatrix}
    w_N \\
    w_{N-1}\\
    \vdots\\
    w_1
  \end{pmatrix}=-
  \begin{pmatrix}
    \D^N(h_1)\\
    \D^N(h_2)\\
    \vdots\\
    \D^N(h_N)
  \end{pmatrix}.
\end{displaymath}
Then the operator $W$ can be written as
\begin{equation}
  \label{eqn:W}
  W = \frac{1
  }{Wrd(h_1,\cdots,h_N)}\left|
    \begin{matrix}
      h_1 & h_2 & \cdots & h_N & 1\\
      \D(h_1) & \D(h_2) & \cdots & \D(h_N) & \D\\
      \vdots & \vdots & \vdots & \vdots & \vdots \\
      \D^N(h_1) & \D^N(h_2) & \cdots & \D^N(h_N)& \D^N
    \end{matrix}\right|
\end{equation}
where $Wrd(h_1,\cdots,h_N)=\left|
    \begin{matrix}
      h_1 & h_2 & \cdots & h_N\\
      \D(h_1) & \D(h_2) & \cdots & \D(h_N)\\
      \vdots & \vdots & \vdots & \vdots\\
      \D^{N-1}(h_1) & \D^{N-1}(h_2) & \cdots & \D^{N-1}(h_N)
    \end{matrix}\right|.$

\begin{prop}
Assume that $h_i$ satisfies
\begin{equation}
\label{eq:hi}
  h_{i,t_n}=\D^n(h_i),~i=1,\cdots,N
\end{equation}
$W$  and $L$ are constructed by (\ref{eqn:W}) and (\ref{eqn:dress}),
then  $W$ and $L$ satisfy (\ref{eqn:W-evolution-dKP}) and
(\ref{eqn:dKP-LaxEqn}), respectively.
\end{prop}
\begin{proof}
 Taking partial derivative $\partial_{t_n}$ to the
equation $W(h_i)=0$:
\begin{displaymath}
  W_{t_n}(h_i)+W\D^n(h_i)=(W_{t_n}+L^n_+W+L^n_-W)(h_i)
\end{displaymath}
$$=(W_{t_n}+L^n_-W)(h_i)=0,~i=1,\cdots, N,$$
since $L^n_-W=L^n W-L^n_+W=W\Delta^{n}-L^n_+W$, $L^n_-W$ is a
non-negative difference operator of order $<N$, $W_{t_n}+L^n_-W$ is
also of order $<N$. Then according to the difference equation's
theory, $W_{t_n}+L^n_-W$ is a zero operator.
 \end{proof}
\subsection{Dressing approach for exDKPH}
We now generalized the dressing approach to exDKPH
(\ref{eqns:exdKP-LaxEqn}). We have
\begin{lemm}
  \label{lm:3}
  Under (\ref{eqn:dress}), if $W$  satisfies
(\ref{eqn:W-evolution-dKP}) and
\begin{equation}
  \label{eq:26}
 W_{\ta_k}=-L^k_-W+\sum_{i=1}^N\psi_i\D^{-1}\phi_iW
\end{equation} then $L$ satisfies (\ref{eqn:exdKP-LaxEqna})
and (\ref{eqn:exdKP-LaxEqnb}).
\end{lemm}
\begin{proof} It is known that $L$ satisfies
(\ref{eqn:exdKP-LaxEqna}). We have
 \begin{align*}
    L_{\ta_k}=&W_{\ta_k}\D W^{-1}-W\D W^{-1} W_{\ta_k} W^{-1}\\
    =&(-L^k_-+\sum_i \psi_i\D^{-1}\phi_i)L+ L(L^k_--\sum_i \psi_i\D^{-1}\phi_i)
    =[B_k+\sum_{i=1}^N \psi_i\D^{-1}\phi_i,L].
  \end{align*}
\end{proof}
 This dressing operator $W$
is constructed as follows: Let $g_i$, $\bar{g}_i$ satisfy
\begin{subequations}
  \label{eq:fg}
  \begin{align}
    &g_{i,t_n}=\D^n(g_i),\quad g_{i,\ta_k}=\D^k(g_i)\quad \\
    &\bar{g}_{i,t_n}=\D^n(\bar{g}_i),\quad \bar{ g}_{i,\ta_k}=\D^k(\bar{g}_i),~i=1,\ldots,N.
  \end{align}
\end{subequations}
And let $h_i$ be the linear combination of $g_i$ and $\bar{g}_i$
\begin{equation}
  \label{eq:h'}
  h_i=g_i+\af_i(\ta_k)\bar{g}_i\quad i=1,\ldots,N,
\end{equation}
with the coefficient $\af_i$ being a differentiable function of
$\ta_k$. Suppose $h_1,\ldots, h_N$ are still linearly independent.

Define
\begin{equation}
  \label{eqn:EigenFns}
 \small{ \psi_i=-\dot{\af}_iW(\bar{g}_i),~
  \phi_i=(-1)^{N-i}\frac{\Wr(\G h_1,\cdots,\hat{\G h}_i,\cdots,\G h_N)}
  {\Wr(\G h_1,\cdots,\G h_N)}, i=1,\ldots, N}
\end{equation}
where the hat $\hat{\;}$ means rule out this term from the discrete
Wronskian determinant, $\dot{\af}_i=\frac{d\af_i}{d\ta_k}$. We have
\begin{prop}
  \label{prop:1}
  Let $W$ be defined by (\ref{eqn:W}) and (\ref{eq:h'}), $L=W\D W^{-1}$, $\psi_{i}$ and $\phi_{i}$
  be given by (\ref{eqn:EigenFns}), then $W$, $L$,
  $\psi_i$, $\phi_i$ satisfy (\ref{eqn:W-evolution-dKP}),(\ref{eq:26}) and  exDKPH (\ref{eqns:exdKP-LaxEqn}).
\end{prop}

To prove  it, we need several lemmas under the above assumptions.
The first one is :
\begin{lemm}{(The  discrete version of Oevel and Strampp's lemma \cite{OS96})}
  \label{lm:OS}
 $$W^{-1}=\sum_{i=1}^N h_i\D^{-1} \phi_i.$$
\end{lemm}
\begin{proof}
  Note that $\phi_1,\ldots,\phi_N$ defined in (\ref{eqn:EigenFns}) satisfy the
  linear equation
 \begin{equation}
   \label{eqn:cond-phi}  \sum_{i=1}^N \D^j(\G h_i) \cdot\phi_i=\dt_{j,N-1}, \quad j=0,1,\cdots, N-1  \end{equation}
where $\dt_{j,N-1}$ is the \emph{Kronecker's delta} symbol. Using
properties  $f\Delta^{-1}=$ \\$\sum_{j\geq
0}\Delta^{-j-1}\Delta^{j}(\Gamma f)$,  we have
  \begin{align*}
    &\sum_{i=1}^{N} h_i\D^{-1}\phi_i=\sum_{i=1}^N\sum_{j=0}^{\infty}\D^{-j-1}\D^j(\G (h_i))\cdot\phi_i
    =\sum_{j=0}^\infty\D^{-j-1}\sum_{i=1}^N\D^j(\G (h_i))\cdot \phi_i\\
    &=\sum_{j=0}^{N-1}\D^{-j-1}\delta_{j,N-1}+\sum_{j=N}^\infty\D^{-j-1}\sum_{i=1}^N\D^j(\G (h_i))\cdot
    \phi_i
    =\Delta^{-N}+O(\Delta^{-N-1}),
  \end{align*}
 So we have
  \begin{equation}
 \label{eq:ad}
    W\sum_ih_i\D^{-1}\phi_i=1+(W\sum_ih_i\D^{-1}\phi_i)_{-}=1+\sum_iW(h_i)\D^{-1}\phi_i=1.
  \end{equation}This complete the proof.
\end{proof}

\begin{lemm}
  \label{lm:vanish}
    $W^*(\phi_i)=0$, for $i=1,\ldots,N$.
\end{lemm}
\begin{proof}
   Lemma \ref{lm:1} implies that
\begin{equation}
 \label{eq:34}
(\Delta^{-1}\phi_i W)_{-}=\Delta^{-1}W^{*}(\phi_i).\end{equation}
Using Lemma \ref{lm:OS} and (\ref{eqn:lemma1a}),
  we have
  \begin{displaymath}
  0=(\D^jW^{-1} W)_{-1}=(\D^j\sum_{i=1}^Nh_i\D^{-1}\phi_iW)_{-}=(\sum_{i=1}^N\D^j(h_{i})\D^{-1}\phi_iW)_{-}
  \end{displaymath}
  $$~~~~ =\sum_{i=1}^N\D^j(h_i)\D^{-1} W^*(\phi_i),~~j=0,\cdots,N-1.$$
 Solving the  equations with respect to
  $\Delta^{-1}W^*(\phi_i)$, we find $\Delta^{-1}W^*(\phi_i)=0.$ This
  implies $W^*(\phi_i)=0.$
  \end{proof}
\begin{lemm}
  \label{lm:key}
  The operator $\D^{-1} \phi_i W$ is a non-negative difference operator
  and
  \begin{equation}
    \label{eq:key-eq}
    (\D^{-1} \phi_i W)(h_j)=\dt_{ij},~1\le i,j\le N.
  \end{equation}
\end{lemm}
\begin{proof}
 Lemma \ref{lm:vanish} and (\ref{eq:34}) implies that
$\Delta^{-1}\phi_{i}W$ is a non-negative difference operator.
   We define functions
  $c_{ij}=(\D^{-1}\phi_iW)(h_j)$,~
 then  $\D(c_{ij})=\phi_iW(h_j)=0$, which means $c_{ij}$ does not depend on
the
  discrete variable $n$. From Lemma \ref{lm:OS}, we find that
  \begin{displaymath}
    \sum_{i=1}^N \D^k(h_i)c_{ij}=\D^k(\sum_i(h_i\D^{-1} \phi_i W)(h_j)) =\D^k(W^{-1}W)(h_j)=\D^k(h_j),
    \end{displaymath}
  so $c_{ij}=\dt_{ij}$.
\end{proof}

{\sl Proof of Proposition \ref{prop:1}.}
  The proof of (\ref{eqn:W-evolution-dKP}) is analogous to the proof of in the previous
  section. For (\ref{eq:26}), taking $\partial_{\ta_k}$ to the identity
  $W(h_i)=0$, using \eqref{eq:fg}, \eqref{eq:h'}, the definition
  (\ref{eqn:EigenFns}) and Lemma \ref{lm:key}, we find
  \begin{align*}
    0=&(W_{\ta_k})(h_i)+(W\D^k)(h_i)+\dot\af_i W(\bar{g}_i)
    =(W_{\ta_k})(h_i)+(L^kW)(h_i)-\sum_{j=1}^N\psi_j\dt_{ji}\\
    =&(W_{\ta_k}+L^k_-W-\sum_{j=1}^N\psi_j\D^{-1}\phi_jW)(h_i).
  \end{align*}
  Since the non-negative difference operator acting on $h_i$ in the last
  expression has degree $<N$, it can not annihilate $N$ independent functions
  unless the operator itself vanishes. Hence (\ref{eq:26}) is
  proved. Then Lemma \ref{lm:3} leads to (\ref{eqn:exdKP-LaxEqnb}).
  The first equation in (\ref{eqn:exdKP-LaxEqnc}) is easy to be verified by a direct calculation, so it remains
  to prove the second equation in (\ref{eqn:exdKP-LaxEqnc}). Firstly, we see that
  $$(W^{-1})_{t_n}=-W^{-1}W_{t_n}W^{-1}=W^{-1}(L^n-B_n)=\D^nW^{-1}-W^{-1}B_n.$$
  Then we substitute $W^{-1}=\sum
  h_i\D^{-1}\phi_i$ to this equality at both ends, we have
  \begin{align*}
    &(W^{-1})_{t_n}=\sum \D^n(h_i)\D^{-1}\phi_i+\sum h_i\D^{-1} \phi_{i,t_n}\\
    &=(\D^nW^{-1}-W^{-1}B_n)_{-}=\sum \D^n(h_i)\D^{-1}\phi_i-\sum
    h_i\D^{-1}B_n^*(\phi_i)
  \end{align*}
  Then $\sum h_i\D^{-1}\phi_{i,t_n}=-\sum h_i\D^{-1}B_n^*(\phi_i)$ implies
  that (\ref{eqn:exdKP-LaxEqnc}) holds.

\section{N-soliton solutions for exDKPH}
Using Proposition \ref{prop:1}, we can find solutions to every
equations in the exDKPH (\ref{eqns:exdKP-LaxEqn}). Let us illustrate
it by solving (\ref{eqns:exam1}) and (\ref{eqns:exam2}). For
(\ref{eqns:exam1}), let $\dt_i=e^{\ld_i}-1$, $\kp_i=e^{\mu_i}-1$, we
take the solution of (\ref{eq:fg}) as follows
\begin{displaymath}
 g_i:=\exp(l\ld_i+\dt_it_1+\dt_i^2\ta_2)=e^{\xi_i}, \quad
 \bar{ g}_i:=\exp(l\mu_i+\kp_it_1+\kp_i^2\ta_2)=e^{\et_i}
\end{displaymath}
\begin{equation}\label{eq:36}
  h_i:=g_i+\af_i(\ta_2)\bar{g}_i=2\sqrt{\af_i}\exp(\frac{\xi_i+\et_i}{2})\cosh(\Om_i),
  ~\Om_i=\frac12(\xi_i-\et_i-\ln \af_i).
\end{equation}
Since $L=W\D W^{-1}=\D+f_0+f_1\D^{-1}+\cdots,$ we have
\begin{equation}\label{eq:f0}
 f_0=Res_{\Delta}(W\D W^{-1}\Delta^{-1})
\end{equation}
where $W$ is given by (\ref{eqn:W}) and (\ref{eq:36}), then
$f_0,~\psi_{i}$ and $\phi_{i}$ given by (\ref{eqn:EigenFns}) gives
rise to the N-soliton solution for (\ref{eqns:exam1}).

For example, we obtain 1-soliton solution for (\ref{eqns:exam1})
with $N=1$ as follows

\begin{displaymath}
  f_0=\exp(\frac{\ld_1+\mu_1}{2})\left(\frac{\cosh(\Om_1+2\tht_1)}{\cosh(\Om_1+\tht_1)}-
    \frac{\cosh(\Om_1+\tht_1)}{\cosh \Om_1}
  \right),\quad \tht_1=\frac{\ld_1-\mu_1}{2}
\end{displaymath}

\begin{displaymath}
  \psi_1=-\frac{d\sqrt{\af_1}}{d\ta_2}(e^{\mu_1-\ld_1})\exp\frac{\xi_1+\et_1}{2}\sech\Om_1,~
  \phi_1=\frac{e^{-(\ld_1+\mu_1)/2}\exp(-\frac{\xi_1+\et_1}{2})}{2\sqrt{\af_1}}\sech(\Om_1+\tht_1).
\end{displaymath}

The 2-soliton solution of (\ref{eqns:exam1}) with $N=2$ is given by
\begin{align*}
  f_0&=-\D(w_1)=(e^{\ld_1}+e^{\ld_2})\D(\frac{v_1}{v}),~
 \\
  \psi_1&=-\frac{\dot\af_1}{v}\left(1+\af_2\frac{(e^{\mu_2}-e^{\ld_1})(e^{\mu_1}-e^{\mu_2})}
    {(e^{\ld_2}-e^{\ld_1})(e^{\mu_1}-e^{\ld_2})}e^{\chi_2}\right)(e^{\mu_1}-e^{\ld_1})(e^{\mu_1}-e^{\ld_2})e^{\et_1},\\
  \psi_2&=-\frac{\dot\af_2}{v}\left(1+\af_1\frac{(e^{\mu_1}-e^{\ld_2})(e^{\mu_1}-e^{\mu_2})}
    {(e^{\ld_2}-e^{\ld_1})(e^{\mu_2}-e^{\ld_1})}e^{\chi_2}\right)(e^{\mu_2}-e^{\ld_2})(e^{\mu_2}-e^{\ld_1})e^{\et_2},\\
  \phi_1&=\G\left(\frac{1+\af_2e^{\chi_2}}{(e^{\ld_1}-e^{\ld_2})v}e^{-\xi_1}\right),~
  \phi_2=\G\left(\frac{1+\af_1e^{\chi_1}}{(e^{\ld_2}-e^{\ld_1})v}e^{-\xi_2}\right),
\end{align*}
with
\begin{align*}
  v&=1+\af_1\frac{e^{\ld_2}-e^{\mu_1}}{e^{\ld_2}-e^{\ld_1}}e^{\chi_1}+
  \af_2\frac{e^{\mu_2}-e^{\ld_1}}{e^{\ld_2}-e^{\ld_1}}e^{\chi_2}+
  \af_1\af_2\frac{e^{\mu_2}-e^{\mu_1}}{e^{\ld_2}-e^{\ld_1}}e^{\chi_1+\chi_2},\\
  v_1&=1+\af_1\frac{e^{2\ld_2}-e^{2\mu_1}}{e^{\ld_2}-e^{\ld_1}}e^{\chi_1}+
  \af_2\frac{e^{2\mu_2}-e^{2\ld_1}}{e^{\ld_2}-e^{\ld_1}}e^{\chi_2}+
  \af_1\af_2\frac{e^{2\mu_2}-e^{2\mu_1}}{e^{\ld_2}-e^{\ld_1}}e^{\chi_1+\chi_2}.
\end{align*}
It can be shown that the interaction between the two solutions is
elastic.

 For (\ref{eqns:exam2}), we take the solution of
(\ref{eq:fg}) as follows
\begin{displaymath}
  g_i:=\exp(l\ld_i+\dt_i\ta_1+\dt_i^2t_2)=e^{\xi_i}, \quad
  \bar{g}_i:=\exp(l\mu_i+\kp_i\ta_1+\kp_i^2t_2)=e^{\et_i}
\end{displaymath}
\begin{displaymath}
  h_i:=g_i+\af_i(\ta_1)\bar{g}_i=2\sqrt{\af_i}\exp(\frac{\xi_i+\et_i}{2})\cosh(\Om_i).
\end{displaymath}
Then $$ f_0=Res_{\Delta}(W\D W^{-1}\Delta^{-1}),~
f_1=Res_{\Delta}(W\D W^{-1})$$ together with $\psi_{i}$ and
$\phi_{i}$ given by (\ref{eqn:EigenFns}) presents the N-soliton
solution for (\ref{eqns:exam2}).

\section*{Acknowledgement}
This work is supported by National Basic Research Program of China
(973 Program) (2007CB814800), China Postdoctoral Science Foundation
funded project (20080430420) and National Natural Science Foundation
of China (10801083,10671121).

\end{document}